\documentclass[manuscript,screen]{acmart}

\setcopyright{none}
\acmJournal{TOCT}
\acmYear{2026}
\acmDOI{}

\usepackage{amsmath,mathtools}
\usepackage{booktabs}
\usepackage{enumitem}
\usepackage{microtype}
\usepackage{array}

\DeclareMathOperator{\Span}{span}

\DeclareMathOperator{\codim}{codim}
\DeclareMathOperator{\Sing}{Sing}
\DeclareMathOperator{\Tor}{Tor}

\DeclareMathOperator{\ABP}{ABP}

\newcommand{\F}{\mathbb{F}}

\newcommand{\Pj}{\mathbb{P}}

\newcommand{\Cb}[1]{\underline{C}_{#1}}
\newcommand{\Dtwo}{\Delta_2}

\theoremstyle{plain}
\newtheorem{theorem}{Theorem}[section]
\newtheorem{lemma}[theorem]{Lemma}
\newtheorem{proposition}[theorem]{Proposition}
\newtheorem{corollary}[theorem]{Corollary}
\theoremstyle{definition}
\newtheorem{definition}[theorem]{Definition}

\begin{document}

\title[Two-Cut Coherence of Quintic Forms]{Two-Cut Coherence of Quintic Forms: Lifting Separations and Second-Derivative Completeness}

\author{Karthik Sheshadri}
\affiliation{%
  \institution{Independent researcher}
  \city{San Jose}
  \state{California}
  \country{USA}}
\email{karthiksheshadri217@gmail.com}
\renewcommand{\shortauthors}{K. Sheshadri}

\begin{abstract}
For a homogeneous polynomial $f$ of degree $d$, the degree-$k$ restricted strength $C_k(f)$ is the least number of products needed in a decomposition with factor degrees $k$ and $d-k$.  We introduce a two-cut coherence parameter
\[
  C_{k,\ell}(f)=\min\left\{r:
  f=\sum_{i,j=1}^{r}p_i m_{ij}q_j,
  \ \deg p_i=k,\ \deg m_{ij}=\ell-k,\ \deg q_j=d-\ell
  \right\},
\]
which requires two interfaces to be realized by one common factorization.  It equals the minimum common endpoint width of a three-block compressed transfer network, or equivalently the minimum, over all tensor lifts of $f$ through commutative multiplication, of the maximum of the two tensor-train endpoint ranks.

Our main result is an extraction-completeness theorem for quintics at cuts $(1,3)$.  Let
\[
  \Dtwo(f)=\max_{u} C_1(\partial_u^2 f)
\]
be the largest polynomial slice rank of a second directional derivative, and put $t=C_3(f)$.  Over an algebraically closed field of characteristic zero,
\[
  \left\lceil\frac{\Dtwo(f)}{3}\right\rceil
  \le \Cb{1,3}(f)
  \le C_{1,3}(f)
  \le t\Dtwo(f)+2t^2.
\]
Thus, when $C_3$ is bounded, ordinary and border two-cut coherence are equivalent up to constants to a one-cut obstruction exposed by a second derivative.

We also prove a border-stable lifting separation.  For coprime nonzero cubics $A,B$, the quintic $L(A,B)=abA+cdB$ has ordinary and border local values $C_1=C_3=2$, while
\[
  \left\lceil\frac{\max\{C_1(A),C_1(B)\}}{3}\right\rceil
  \le \Cb{1,3}(L)
  \le C_1(A)+C_1(B).
\]
Taking $A=\sum_i v_i^3$ gives an unbounded gap between separately optimal local interfaces and a common interface, even in border complexity.
\end{abstract}




\maketitle

\section{Introduction}\label{sec:intro}

A decomposition of a homogeneous polynomial across one degree cut gives a local description.  If $f\in R_d$ and $0<k<d$, the restricted strength
\[
  C_k(f)=\min\left\{r:f=\sum_{i=1}^r g_i h_i,
  \ \deg g_i=k,\ \deg h_i=d-k\right\}
\]
asks for the smallest interface at degree $k$.  Restricted strength is a well-established refinement of polynomial strength and, at $k=1$, is polynomial slice rank; it has also been used to lower-bound layers of homogeneous algebraic branching programs (ABPs)~\cite{AnanyanHochster2020,BikDraismaEggermont2019,GesmundoGhosalIkenmeyerLysikov2022,FlaviGesmundoOnetoVentura2025}.

Two separately small interfaces need not arise from one common computation.  Fix $0<k<\ell<d$.  We study the least $r$ for which
\begin{equation}\label{eq:twocut-intro}
  f=\sum_{i,j=1}^r p_i m_{ij}q_j,
  \qquad
  p_i\in R_k,\quad m_{ij}\in R_{\ell-k},\quad q_j\in R_{d-\ell}.
\end{equation}
We call this quantity the \emph{two-cut coherence complexity} $C_{k,\ell}(f)$.  Regrouping at either endpoint gives
\[
  C_{k,\ell}(f)\ge \max\{C_k(f),C_\ell(f)\},
\]
but the inequality can be highly non-tight.  The central issue is whether this extra cost is genuinely new or simply conceals an ordinary one-cut obstruction.

\subsection{Main theorem}

For quintics and cuts $(1,3)$, the obstruction is completely detected, up to constants at bounded $C_3$, by second derivatives.  For a directional derivative $\partial_u$, define
\begin{equation}\label{eq:Delta-def-intro}
  \Dtwo(f)=\max_{u} C_1(\partial_u^2 f).
\end{equation}
Our main theorem is the following dimension-free sandwich.

\begin{theorem}[Second-derivative completeness]\label{thm:main-intro}
Let $f$ be a nonzero quintic over an algebraically closed field of characteristic zero, and let $t=C_3(f)$.  Then
\begin{equation}\label{eq:main-sandwich-intro}
  \boxed{
  \left\lceil\frac{\Dtwo(f)}{3}\right\rceil
  \le \Cb{1,3}(f)
  \le C_{1,3}(f)
  \le t\Dtwo(f)+2t^2.}
\end{equation}
\end{theorem}

The left inequality follows from a three-group Leibniz identity: a width-$r$ representation of the form~\eqref{eq:twocut-intro} sends every square second derivative to a cubic of slice rank at most $3r$.  The right inequality starts with a minimal degree-$3$ decomposition
\[
  f=\sum_{i=1}^t G_iQ_i,
  \qquad G_i\in R_3,\quad Q_i\in R_2.
\]
Square directional derivatives of quadrics are point evaluations.  One can therefore choose $t$ directions dual to the quadratic endpoint space $\Span\{Q_i\}$.  Each derivative isolates one coefficient cubic $G_i$, up to a Leibniz error of slice rank at most $2t$.  Decomposing the $G_i$ into slices then builds a common two-cut network.

Since every nonzero quintic has $\Dtwo(f)\ge1$, Theorem~\ref{thm:main-intro} implies that for every fixed $t_0$,
\begin{equation}\label{eq:bounded-equivalence-intro}
  \frac{\Dtwo(f)}{3}
  \le \Cb{1,3}(f)
  \le C_{1,3}(f)
  \le (t_0+2t_0^2)\Dtwo(f)
  \qquad\text{whenever }C_3(f)\le t_0.
\end{equation}
Thus no family of quintics can simultaneously have bounded $C_3$, unbounded border two-cut coherence, and second-derivative slice complexity negligible compared with the joint width.

\subsection{A border-stable separation}

The completeness theorem does not collapse two-cut coherence to the maximum of its local widths.  For nonzero coprime cubics $A,B$, set
\[
  L(A,B)=abA+cdB.
\]
The two local descriptions
\[
  L=a(bA)+c(dB),
  \qquad
  L=A(ab)+B(cd)
\]
have width two, and coprimality makes $L$ irreducible.  Nevertheless,
\begin{equation}\label{eq:lifting-intro}
  \left\lceil\frac{\max\{C_1(A),C_1(B)\}}{3}\right\rceil
  \le \Cb{1,3}(L)
  \le C_{1,3}(L)
  \le C_1(A)+C_1(B).
\end{equation}
For the Fermat cubic $A_n(v)=\sum_{i=1}^n v_i^3$, we prove $C_1(A_n)=\lceil n/2\rceil$.  Hence
\[
  F_n=abA_n+cdv_1^3
\]
has local ordinary and border values $2,2$, while its common-interface width is $\Omega(n)$, including in border complexity.  This refutes any universal additive bound of the form $C_{k,\ell}\le C_k+C_\ell$.

\subsection{Computational meaning}

The quantity $C_{k,\ell}$ is not standard homogeneous ABP width.  In~\eqref{eq:twocut-intro}, the middle labels $m_{ij}$ are arbitrary homogeneous polynomials and their internal complexity is free.  We prove that $C_{k,\ell}$ is exactly the minimum interface width of a three-block compressed transfer network.  Equivalently, if
\[
  \mu:R_k\otimes R_{\ell-k}\otimes R_{d-\ell}\longrightarrow R_d
\]
is commutative multiplication, then $C_{k,\ell}(f)$ is the minimum endpoint tensor-train rank among all lifts in $\mu^{-1}(f)$.  Every standard homogeneous ABP yields such a lift, so $C_{k,\ell}$ is always an ABP-width lower bound.  For quintics at cuts $(1,3)$, the hidden quadratic labels can be compiled with polynomial overhead.

\subsection{Contributions and limitations}

The paper makes four contributions.
\begin{enumerate}[leftmargin=2em]
  \item It defines two-cut coherence and identifies its exact compressed-network and multiplication-fiber tensor interpretations.
  \item It proves the border-stable lifting theorem~\eqref{eq:lifting-intro}, giving an unbounded gap between separately optimized local interfaces and one common interface.
  \item It proves the sharp three-group second-derivative transfer lemma.
  \item It proves the universal extraction-completeness inequality~\eqref{eq:main-sandwich-intro} for quintics at cuts $(1,3)$.
\end{enumerate}

The results are deliberately narrower than a general circuit lower bound.  The separating family has an $O(n)$-size arithmetic formula, and the completeness theorem is specific to degree five and the cut profile $(1,2,2)$.  Nothing here separates $\mathsf{VP}$ from $\mathsf{VNP}$, determinant from permanent, or polynomial from superpolynomial computation.

\paragraph{Research process and verification.}
Generative AI systems, including ChatGPT, were used during the exploratory research process for candidate generation, adversarial criticism, proof drafting, symbolic-verification scaffolding, and exposition.  The author is responsible for the final mathematical statements, proofs, computations, citations, and submission decisions.  AI systems are not listed as authors or used as citable sources.  All mathematical claims are intended to be independently checkable from the paper; exact symbolic verification scripts are included as supplementary material, but no lower bound in the paper depends on a numerical or heuristic computation.

\section{Definitions and elementary properties}\label{sec:prelim}

Throughout, $\F$ is an algebraically closed field of characteristic zero, $V$ is a finite-dimensional vector space, and
\[
  R=\operatorname{Sym}(V^*)=\bigoplus_{d\ge0}R_d
\]
is the standard graded polynomial ring.  The zero polynomial is assigned complexity zero.

\begin{definition}[One-cut and two-cut complexity]\label{def:complexities}
For $f\in R_d$ and $0<k<d$, define
\[
  C_k(f)=\min\left\{r:f=\sum_{i=1}^r p_iq_i,
  \ p_i\in R_k,\ q_i\in R_{d-k}\right\}.
\]
For $0<k<\ell<d$, define
\[
  C_{k,\ell}(f)=\min\left\{r:f=\sum_{i,j=1}^r p_im_{ij}q_j,
  \ p_i\in R_k,\ m_{ij}\in R_{\ell-k},\ q_j\in R_{d-\ell}\right\}.
\]
Their border versions $\Cb{k}$ and $\Cb{k,\ell}$ are the least $r$ for which $f$ lies in the Zariski closure of the corresponding complexity-at-most-$r$ locus.
\end{definition}

\begin{proposition}[Endpoint-subspace formulation]\label{prop:endpoint-subspace}
For $f\in R_d$,
\begin{align}
  C_k(f)
  &=\min\{\dim P:P\le R_k,\ f\in P R_{d-k}\},\label{eq:one-subspace}\\
  C_{k,\ell}(f)
  &=\min\{\max(\dim P,\dim Q):
  P\le R_k,\ Q\le R_{d-\ell},\ f\in P R_{\ell-k}Q\}.
  \label{eq:two-subspace}
\end{align}
Consequently,
\begin{equation}\label{eq:local-lower}
  C_{k,\ell}(f)\ge \max\{C_k(f),C_\ell(f)\}.
\end{equation}
\end{proposition}

\begin{proof}
A product decomposition supplies the spans of its endpoint factors.  Conversely, after choosing bases of the endpoint spaces, one collects coefficients into the remaining factor or into the middle matrix.  Regrouping a two-cut representation first by $i$ and then by $j$ gives one-cut representations at $k$ and $\ell$, proving~\eqref{eq:local-lower}.
\end{proof}

For $C_1$, there is a useful geometric interpretation.

\begin{proposition}[Linear spaces and closedness of polynomial slice rank]\label{prop:slice-geometry}
Let $g\in R_d$ be nonzero.  Then $C_1(g)$ is the minimum codimension of a linear subspace $W\le V$ on which $g$ vanishes identically.  For each $s$, the projective locus
\[
  \{[g]\in\Pj(R_d):C_1(g)\le s\}
\]
is Zariski closed.
\end{proposition}

\begin{proof}
If $g=\sum_{i=1}^s\ell_i h_i$, then $g$ vanishes on the common kernel of the $\ell_i$, a subspace of codimension at most $s$.  Conversely, the homogeneous ideal of a codimension-$s$ linear subspace is generated by $s$ independent linear forms, so every form vanishing on it belongs to their ideal.

For closedness, consider the incidence variety of pairs $([g],W)$ with $W$ a codimension-at-most-$s$ linear subspace and $g|_W=0$.  It is closed over a projective Grassmannian, hence its projection to $\Pj(R_d)$ is closed.
\end{proof}

For $u\in V$, let $\partial_u$ denote the constant-coefficient directional derivative.  If $f\in R_5$, define
\begin{equation}\label{eq:Delta-def}
  \Dtwo(f)=\max_{u\in V} C_1(\partial_u^2f).
\end{equation}
The maximum exists because the values are integers in a finite range.  Scaling a nonzero direction does not change the slice rank.

\begin{lemma}\label{lem:Delta-positive}
Every nonzero quintic satisfies $\Dtwo(f)\ge1$.
\end{lemma}

\begin{proof}
If $\partial_u^2f=0$ for all $u$, polarization implies that every second partial derivative of $f$ vanishes.  In characteristic zero this forces $f$ to have degree at most one, contradicting that it is a nonzero quintic.
\end{proof}

\section{The exact computational model}\label{sec:model}

\subsection{Compressed transfer networks}

\begin{definition}[Three-block compressed transfer network]\label{def:ctn}
A three-block homogeneous compressed transfer network of degree profile
$(k,\ell-k,d-\ell)$ consists of a left interface with labels $p_i\in R_k$, a right interface with labels $q_j\in R_{d-\ell}$, and arbitrary transfer labels $m_{ij}\in R_{\ell-k}$.  It computes
\[
  \sum_{i=1}^{r_L}\sum_{j=1}^{r_R}p_im_{ij}q_j.
\]
Its interface width is $\max\{r_L,r_R\}$.  The internal circuit complexity of the transfer labels is not charged.
\end{definition}

\begin{theorem}[Exact compressed-network interpretation]\label{thm:ctn-equivalence}
For every $f\in R_d$,
\[
  C_{k,\ell}(f)=
  \min\{\text{interface width of a compressed transfer network computing }f\}.
\]
\end{theorem}

\begin{proof}
Definition~\ref{def:complexities} is exactly the square-interface case.  A rectangular network of width $r$ can be padded by zero labels to an $r\times r$ representation, and every square representation is a network.
\end{proof}

\subsection{Minimum tensor-train rank in a multiplication fiber}

Put
\[
  U=R_k,\qquad W=R_{\ell-k},\qquad Z=R_{d-\ell},
\]
and let
\[
  \mu:U\otimes W\otimes Z\longrightarrow R_d
\]
be commutative multiplication.  For $T\in U\otimes W\otimes Z$, let $\rho_U(T)$ and $\rho_Z(T)$ be the ranks of the two endpoint flattenings.  These are the two bond ranks of an order-three tensor train~\cite{Oseledets2011,YeLim2018}.

\begin{theorem}[Multiplication-fiber tensor-train characterization]\label{thm:fiber-tt}
For every $f\in R_d$,
\begin{equation}\label{eq:fiber-tt}
  C_{k,\ell}(f)
  =\min_{T\in\mu^{-1}(f)}\max\{\rho_U(T),\rho_Z(T)\}.
\end{equation}
\end{theorem}

\begin{proof}
A width-$r$ representation gives the lift
\[
  T=\sum_{i,j=1}^r p_i\otimes m_{ij}\otimes q_j,
\]
whose endpoint supports have dimension at most $r$.  Conversely, if both endpoint flattening ranks of a lift are at most $r$, then
\[
  T\in P\otimes W\otimes Q
\]
for endpoint spaces $P\le U$ and $Q\le Z$ of dimension at most $r$.  Expanding in bases of $P$ and $Q$ and applying $\mu$ yields a width-$r$ two-cut representation.
\end{proof}

Equation~\eqref{eq:fiber-tt} explains the role of commutativity.  For a fixed tensor, path ranks are determined by flattenings.  For a commutative polynomial, there are many tensor lifts in the affine fiber $\mu^{-1}(f)$, and different lifts may optimize different cuts.  The common problem minimizes both endpoint ranks on one lift.

\subsection{Relationship with homogeneous ABPs}

A standard homogeneous single-source/single-sink ABP is layered by degree and has linear edge labels.

\begin{theorem}[ABP lower-bound direction]\label{thm:abp-lower}
Suppose a homogeneous ABP computing $f$ has widths $w_k$ and $w_\ell$ at layers $k$ and $\ell$.  Then
\[
  C_{k,\ell}(f)\le\max\{w_k,w_\ell\}.
\]
In particular, $C_{k,\ell}(f)$ lower-bounds the minimum homogeneous ABP width of $f$.
\end{theorem}

\begin{proof}
Let $p_i$ be the source-to-layer-$k$ path polynomial, $m_{ij}$ the layer-$k$-to-layer-$\ell$ path polynomial, and $q_j$ the layer-$\ell$-to-sink path polynomial.  Path concatenation gives $f=\sum_{i,j}p_im_{ij}q_j$ with the required degrees.  Pad the smaller layer by zero nodes.
\end{proof}

The converse is not equality because the $m_{ij}$ may conceal expensive computations.  At fixed degree, however, they can be compiled.

\begin{proposition}[Quintic compilation]\label{prop:quintic-compile}
Let $f$ be a quintic in $N$ variables and let $r=C_{1,3}(f)$.  If $w_{\ABP}(f)$ is minimum homogeneous ABP width, then
\begin{equation}\label{eq:abp-poly}
  C_{1,3}(f)\le w_{\ABP}(f)\le Nr^2.
\end{equation}
\end{proposition}

\begin{proof}
Only the upper bound needs proof.  Write each quadratic transfer label as
\[
  m_{ij}=\sum_{s=1}^N x_sL_{ijs}
\]
with $L_{ijs}$ linear.  Between layers one and three, use layer-two nodes indexed by $(i,j,s)$, with edges labeled $x_s$ and $L_{ijs}$.  This requires at most $Nr^2$ layer-two nodes.  Similarly write each quadratic right label as $q_j=\sum_s x_sK_{js}$ and use at most $Nr$ layer-four nodes.  The remaining layers have width at most $r$.
\end{proof}

The standard noncommutative and ordered set-multilinear settings do not have this multiplication-fiber ambiguity: ordered multiplication is injective, and layer ranks are characterized by coefficient-matrix or flattening ranks~\cite{Nisan1991,BlaeserIkenmeyerMahajanPandeySaurabh2020}.

\section{Second-derivative transfer}\label{sec:transfer}

The next identity is the source of the constant three.

\begin{lemma}[Three-group Leibniz identity]\label{lem:three-group}
Let
\[
  H=p^{\mathsf T}Mq
\]
be a quintic width-$r$ representation at cuts $(1,3)$, where $p$ is an $r$-vector of linear forms and $M,q$ have quadratic entries.  For every directional derivative $\delta$,
\begin{align}
  \delta^2H
  ={}&(2p_\delta^{\mathsf T}M_\delta+p^{\mathsf T}M_{\delta\delta})q\notag\\
  &+(2p_\delta^{\mathsf T}M+2p^{\mathsf T}M_\delta)q_\delta
  +p^{\mathsf T}Mq_{\delta\delta}.
  \label{eq:three-group}
\end{align}
Each line of~\eqref{eq:three-group} has polynomial slice rank at most $r$.  Hence
\begin{equation}\label{eq:derivative-transfer}
  C_1(\delta^2H)\le3r.
\end{equation}
\end{lemma}

\begin{proof}
Apply the product rule twice.  In the first line, group by the $r$ quadratic entries of $q$; their coefficients are linear.  In the second, group by the $r$ linear entries of $q_\delta$; their coefficients are quadratic.  In the third, $q_{\delta\delta}$ is scalar-valued, so group by the $r$ linear entries of $p$.
\end{proof}

\begin{corollary}[Border-stable transfer]\label{cor:border-transfer}
For every quintic $f$,
\begin{equation}\label{eq:Delta-lower}
  \Dtwo(f)\le3\Cb{1,3}(f).
\end{equation}
Equivalently,
\[
  \left\lceil\frac{\Dtwo(f)}3\right\rceil\le\Cb{1,3}(f).
\]
\end{corollary}

\begin{proof}
The exact statement follows from Lemma~\ref{lem:three-group}.  If $f$ is a limit of width-$r$ quintics, differentiation carries the degeneration to a limit of cubics of slice rank at most $3r$.  Proposition~\ref{prop:slice-geometry} says that the cubic slice-rank-at-most-$3r$ locus is closed, so the limit has the same bound.  Maximize over directions.
\end{proof}

For two distinct derivatives, the Leibniz expansion groups into four slices-per-interface rather than three.  The square identity is therefore stronger for the present purpose.  Appendix~\ref{app:sharpness} proves that the factor three is asymptotically optimal for a black-box transfer lemma applied to arbitrary $p^{\mathsf T}Mq$ representations: there are width-$r$ examples whose extracted cubic has slice rank $3r-1$.

\section{A border-stable lifting separation}\label{sec:lifting}

\begin{theorem}[Cubic lifting theorem]\label{thm:lifting}
Let $A,B\in\F[v_1,\ldots,v_n]_3$ be nonzero cubics with $\gcd(A,B)=1$, and set
\[
  L(A,B)=abA(v)+cdB(v).
\]
At cuts $(1,3)$,
\begin{equation}\label{eq:local-values}
  C_1(L)=\Cb{1}(L)=C_3(L)=\Cb{3}(L)=2,
\end{equation}
and
\begin{equation}\label{eq:lifting-main}
  \left\lceil\frac{\max\{C_1(A),C_1(B)\}}3\right\rceil
  \le\Cb{1,3}(L)
  \le C_{1,3}(L)
  \le C_1(A)+C_1(B).
\end{equation}
\end{theorem}

\begin{proof}
The local upper bounds are
\[
  L=a(bA)+c(dB),
  \qquad
  L=A(ab)+B(cd).
\]
For the lower bounds, view $L$ as a polynomial in $a$ over
\[
  S=\F[b,c,d,v_1,\ldots,v_n].
\]
Then
\[
  L=(bA)a+cdB.
\]
The coefficients $bA$ and $cdB$ are coprime, so $L$ is primitive in $S[a]$.  It is linear and nonconstant in $a$, hence irreducible over $\operatorname{Frac}(S)$ and therefore irreducible in $S[a]$ by Gauss's lemma.  A value-one $C_1$ or $C_3$ representation would be a nontrivial factorization of degrees $1+4$ or $3+2$, respectively.  Thus both ordinary local values are two.

The rank-one product loci for the splits $1+4$ and $3+2$ are projective images of products of projective spaces, hence closed.  Therefore the local border values are also two.

For the common-interface upper bound, choose optimal decompositions
\[
  A=\sum_{i=1}^{r_A}\lambda_iQ_i,
  \qquad
  B=\sum_{j=1}^{r_B}\mu_jS_j,
\]
where $r_A=C_1(A)$ and $r_B=C_1(B)$.  Put the terms on the diagonal of an $(r_A+r_B)\times(r_A+r_B)$ transfer matrix: use
\[
  (p_i,m_{ii},q_i)=(a,Q_i,b\lambda_i)
\]
for the $A$-terms and
\[
  (p_{r_A+j},m_{r_A+j,r_A+j},q_{r_A+j})=(c,S_j,d\mu_j)
\]
for the $B$-terms.  All off-diagonal entries are zero.  This gives $L$ and proves the upper bound.

Finally,
\[
  (\partial_a+\partial_b)^2L=2A,
  \qquad
  (\partial_c+\partial_d)^2L=2B.
\]
Hence $\Dtwo(L)\ge\max\{C_1(A),C_1(B)\}$.  Corollary~\ref{cor:border-transfer} proves the lower bound in~\eqref{eq:lifting-main}.
\end{proof}

The coprimality hypothesis is exact for this presentation: if $A$ and $B$ have a common nonconstant factor, it divides $L$.

\subsection{The Fermat cubic}

\begin{proposition}\label{prop:Fermat-slice}
Let
\[
  A_n(v)=\sum_{i=1}^n v_i^3.
\]
Then
\begin{equation}\label{eq:Fermat-slice}
  C_1(A_n)=\left\lceil\frac n2\right\rceil.
\end{equation}
\end{proposition}

\begin{proof}
For the upper bound, impose $v_{2j}=-v_{2j-1}$ for each pair, and if $n$ is odd also impose $v_n=0$.  This gives a linear subspace of dimension $\lfloor n/2\rfloor$ on which $A_n$ vanishes.  Proposition~\ref{prop:slice-geometry} gives $C_1(A_n)\le\lceil n/2\rceil$.

For the reverse inequality, let $W\le\F^n$ be a linear subspace on which $A_n$ vanishes.  Choose $w\in W$ with maximal coordinate support $S$.  Every vector of $W$ is supported in $S$: otherwise a generic linear combination with $w$ would have strictly larger support.  Polarizing the identity $A_n|_W=0$ gives
\[
  \sum_{i\in S}u_iv_iw_i=0
  \qquad\text{for all }u,v\in W.
\]
Because $w_i\ne0$ for $i\in S$, the bilinear form
\[
  B_w(u,v)=\sum_{i\in S}w_i u_iv_i
\]
is nondegenerate on $\F^S$.  Thus $W\subseteq W^{\perp_{B_w}}$, so
\[
  2\dim W\le |S|\le n.
\]
Therefore every linear subspace contained in $A_n=0$ has dimension at most $\lfloor n/2\rfloor$.  Apply Proposition~\ref{prop:slice-geometry}.
\end{proof}

\begin{corollary}[Unbounded local-to-common gap]\label{cor:Fermat-lift}
For $n\ge2$, define
\[
  F_n=ab\sum_{i=1}^n v_i^3+cdv_1^3.
\]
Then
\[
  C_1(F_n)=\Cb{1}(F_n)=C_3(F_n)=\Cb{3}(F_n)=2,
\]
while
\begin{equation}\label{eq:Fermat-lift-lb}
  \Cb{1,3}(F_n)
  \ge
  \left\lceil\frac{\lceil n/2\rceil}{3}\right\rceil.
\end{equation}
Thus the ratio between common two-cut width and the maximum local width tends to infinity.  Moreover, $\Cb{1,3}(F_n)>C_1(F_n)+C_3(F_n)$ for $n\ge25$.
\end{corollary}

\begin{proof}
The cubics $A_n$ and $v_1^3$ are coprime for $n\ge2$.  Apply Theorem~\ref{thm:lifting} and Proposition~\ref{prop:Fermat-slice}.
\end{proof}

By Theorem~\ref{thm:abp-lower}, every homogeneous ABP for $F_n$ has width $\Omega(n)$.  This is compatible with the fact that $F_n$ has an $O(n)$-size arithmetic formula: width and formula size are different resources, and the lower bound is only linear.

\section{Second-derivative completeness}\label{sec:completeness}

We now prove the upper bound in Theorem~\ref{thm:main-intro}.

\begin{lemma}[Quadratic evaluations span the dual]\label{lem:eval-dual}
Let $W\le R_2$ be a $t$-dimensional space of quadrics over an infinite field.  The functionals
\[
  e_u:W\to\F,
  \qquad e_u(Q)=Q(u),
\]
span $W^*$.
\end{lemma}

\begin{proof}
If their span were proper, some nonzero $Q\in W$ would be annihilated by every $e_u$.  Then $Q(u)=0$ for all $u$, which is impossible for a nonzero polynomial over an infinite field.
\end{proof}

For a homogeneous quadratic $Q$, $\partial_u^2Q=2Q(u)$.  Thus Lemma~\ref{lem:eval-dual} also says that square second-order differential evaluations span the dual of every finite-dimensional quadratic endpoint space.

\begin{theorem}[Universal completeness theorem]\label{thm:completeness}
Let $f\in R_5$ be nonzero and put $t=C_3(f)$.  Then
\begin{equation}\label{eq:completeness-upper}
  C_{1,3}(f)\le t\Dtwo(f)+2t^2.
\end{equation}
Consequently,
\begin{equation}\label{eq:main-sandwich}
  \left\lceil\frac{\Dtwo(f)}3\right\rceil
  \le\Cb{1,3}(f)
  \le C_{1,3}(f)
  \le C_3(f)\Dtwo(f)+2C_3(f)^2.
\end{equation}
\end{theorem}

\begin{proof}
Choose a minimal degree-$3$ decomposition
\begin{equation}\label{eq:min-C3}
  f=\sum_{i=1}^tG_iQ_i,
  \qquad G_i\in R_3,\quad Q_i\in R_2.
\end{equation}
Minimality implies that the $Q_i$ are linearly independent.  Let $W=\Span\{Q_1,\ldots,Q_t\}$.  By Lemma~\ref{lem:eval-dual}, choose directions $u_1,\ldots,u_t$ such that the functionals $Q\mapsto\partial_{u_j}^2Q$ form a basis of $W^*$.  Change the basis of the $Q_i$ and make the inverse change to the coefficient cubics so that
\begin{equation}\label{eq:dual-normalization}
  \partial_{u_j}^2Q_i=\delta_{ij}.
\end{equation}

Differentiate~\eqref{eq:min-C3} twice in direction $u_j$.  The product rule and~\eqref{eq:dual-normalization} give
\begin{equation}\label{eq:isolate-G}
  \partial_{u_j}^2f
  =G_j+
  \sum_{i=1}^t\left[
    2(\partial_{u_j}G_i)(\partial_{u_j}Q_i)
    +(\partial_{u_j}^2G_i)Q_i
  \right].
\end{equation}
Each term in brackets is a product of a linear and a quadratic form.  Hence the error in~\eqref{eq:isolate-G} has slice rank at most $2t$, and subadditivity gives
\begin{equation}\label{eq:G-bound}
  C_1(G_j)\le C_1(\partial_{u_j}^2f)+2t
  \le \Dtwo(f)+2t.
\end{equation}

Choose optimal slice decompositions
\[
  G_j=\sum_{h=1}^{s_j}\ell_{jh}A_{jh},
  \qquad s_j=C_1(G_j),
\]
with $\ell_{jh}$ linear and $A_{jh}$ quadratic.  Substitution into~\eqref{eq:min-C3} yields
\[
  f=\sum_{j=1}^t\sum_{h=1}^{s_j}\ell_{jh}A_{jh}Q_j.
\]
Putting these products on the diagonal of a compressed transfer network gives
\[
  C_{1,3}(f)\le\sum_{j=1}^t C_1(G_j).
\]
Using~\eqref{eq:G-bound},
\[
  C_{1,3}(f)
  \le t(\Dtwo(f)+2t)
  =t\Dtwo(f)+2t^2.
\]
The remaining inequalities in~\eqref{eq:main-sandwich} are Corollary~\ref{cor:border-transfer} and the definition of border complexity.
\end{proof}

\begin{corollary}[Equivalence at bounded local width]\label{cor:bounded-equivalence}
For every fixed $t_0$, if $C_3(f)\le t_0$ then
\begin{equation}\label{eq:bounded-equivalence}
  \frac{\Dtwo(f)}3
  \le\Cb{1,3}(f)
  \le C_{1,3}(f)
  \le (t_0+2t_0^2)\Dtwo(f).
\end{equation}
\end{corollary}

\begin{proof}
Apply Theorem~\ref{thm:completeness} and Lemma~\ref{lem:Delta-positive}.
\end{proof}

\begin{corollary}[No extraction-invisible family at bounded $C_3$]\label{cor:no-intrinsic}
There is no family of quintics $(f_n)$ with
\[
  C_3(f_n)=O(1),
  \qquad
  \Cb{1,3}(f_n)\to\infty,
  \qquad
  \Dtwo(f_n)=o(\Cb{1,3}(f_n)).
\]
The conclusion remains true if one additionally requires $C_1(f_n)=O(1)$.
\end{corollary}

The theorem has a concrete interpretation.  An optimal degree-$3$ endpoint space consists of quadrics.  Quadrics are exactly dual to square second-order differential evaluations.  A bounded set of such contractions therefore reveals every hidden cubic coefficient, with a controlled Leibniz error.  The common-interface obstruction is real, as the lifting family shows, but at bounded $C_3$ it cannot remain invisible to second derivatives.

\section{Fixed-endpoint geometry and Tor}\label{sec:tor}

Two-cut coherence also has a useful commutative-algebra interpretation.  Fix endpoint spaces
\[
  P\le R_k,
  \qquad
  Q\le R_{d-\ell},
\]
and let $I=(P)$ and $J=(Q)$ be their homogeneous ideals.  Separate compatibility with the two endpoints means
\[
  f\in I_d\cap J_d,
\]
whereas a common realization with these same endpoints means
\[
  f\in(IJ)_d=P R_{\ell-k}Q.
\]
The obstruction is therefore the class
\begin{equation}\label{eq:tor-class}
  [f]_{P,Q}\in\left(\frac{I\cap J}{IJ}\right)_d.
\end{equation}

\begin{proposition}\label{prop:tor}
There is a natural graded isomorphism
\begin{equation}\label{eq:tor-iso}
  \frac{I\cap J}{IJ}
  \cong\Tor_1^R(R/I,R/J).
\end{equation}
\end{proposition}

\begin{proof}
Tensor the short exact sequence $0\to I\to R\to R/I\to0$ with $R/J$.  The kernel of $I/IJ\to R/J$ is $(I\cap J)/IJ$, and exactness identifies it with $\Tor_1^R(R/I,R/J)$; see, for example,~\cite{Eisenbud1995}.
\end{proof}

Thus fixed-endpoint tension is a failure of transversality.  If generators of $J$ form a regular sequence and remain regular modulo $I$, the Koszul resolution remains exact after tensoring with $R/I$, so the Tor group vanishes and separate compatibility glues.  Theorem~\ref{thm:completeness} adds a quantitative statement in the quintic $(1,3)$ regime: when the quadratic endpoint dimension is bounded, any large failure of gluing produces a large cubic slice obstruction under a square contraction.

\section{Comparison with related notions}\label{sec:related}

The ambient ingredients of this paper are established: polynomial strength and slice rank, restricted strength, tensor-train rank, and homogeneous ABPs.  The new point is the shared-interface condition after quotienting by commutative multiplication.

\subsection{Restricted strength and slice rank}

The quantity $C_k$ is exactly degree-$k$ restricted strength.  At $k=1$ it is polynomial slice rank, geometrically the minimum codimension of a linear space contained in the hypersurface.  These notions originate in the strength framework around Stillman's conjecture and have since been studied algebraically and geometrically~\cite{AnanyanHochster2020,BikDraismaEggermont2019,BallicoVentura2021,GesmundoGhosalIkenmeyerLysikov2022,FlaviGesmundoOnetoVentura2025}.  Polynomial slice rank should not be confused with the slice rank of a fixed multilinear tensor, whose direct-sum and diagonal behavior is a separate theory~\cite{Gowers2021}.  The inequality $C_{k,\ell}\ge\max\{C_k,C_\ell\}$ is immediate, while Corollary~\ref{cor:Fermat-lift} shows that the gap can be unbounded and border-stable.

\subsection{Three-factor product rank}

For $\lambda=(k,\ell-k,d-\ell)$, let $S_\lambda(f)$ be the least number of summands in a decomposition of $f$ into products of three forms of the prescribed degrees.  This is the corresponding three-factor $\lambda$-strength; recent work relates such partition-rank variants to multiplicative complexity~\cite{BrandKaskiWang2026}.

\begin{proposition}\label{prop:lambda-strength}
For every $f$,
\begin{equation}\label{eq:lambda-strength}
  C_{k,\ell}(f)\le S_\lambda(f)\le C_{k,\ell}(f)^2.
\end{equation}
\end{proposition}

\begin{proof}
A sum of $s$ independent three-factor products embeds diagonally in a width-$s$ two-cut representation.  Conversely, expanding a width-$r$ transfer matrix produces at most $r^2$ three-factor products.
\end{proof}

Thus three-factor product rank counts independent paths, whereas two-cut coherence charges the dimensions of shared endpoint spaces.

\subsection{Tensor trains and ABPs}

For a fixed order-three tensor, tensor-train ranks are endpoint unfolding ranks~\cite{Oseledets2011}.  General graph tensor-network ranks place this in a broader geometric framework~\cite{YeLim2018}.  Theorem~\ref{thm:fiber-tt} differs in one essential respect: it minimizes path rank over the multiplication fiber $\mu^{-1}(f)$.  It is therefore not the tensor-train rank of a canonical polarization tensor.

Theorem~\ref{thm:abp-lower} gives the safe relationship with standard homogeneous ABPs: two-cut coherence is a lower bound on width.  The reverse direction must charge the hidden transfer labels, and Proposition~\ref{prop:quintic-compile} gives one fixed-degree compilation.  This distinction is consistent with the established rank characterizations in noncommutative models and the geometric study of ABP-width loci~\cite{Nisan1991,BlaeserIkenmeyerMahajanPandeySaurabh2020}.

\subsection{Novelty boundary}

To the author's knowledge, the precise shared-interface parameter in Definition~\ref{def:complexities}, its multiplication-fiber formulation, the border-stable lifting separation, and Theorem~\ref{thm:completeness} do not appear in the cited literature.  This is a conservative claim about the theorems, not a claim that projected subspace varieties or path tensor networks are new constructions.  In particular, the terms ``two-cut coherence'' and ``multiplication-fiber tensor-train width'' are descriptive terminology for the invariant studied here.

\section{Sharpness, scope, and consequences}\label{sec:scope}

The factor three in Corollary~\ref{cor:border-transfer} is asymptotically optimal for a universal black-box derivative transfer from arbitrary width-$r$ representations: Appendix~\ref{app:sharpness} constructs $H_r=p^{\mathsf T}Mq$ with
\[
  C_1\left(\frac12\partial_\tau^2H_r\big|_{\tau=0}\right)=3r-1.
\]
This does not rule out a better constant on the special support locus $abA+cdB$; such an improvement would need to use those support equations, not only the Leibniz identity.

The extraction-completeness theorem has an exact boundary.  It concerns degree five, cuts $(1,3)$, and ordinary $C_3$ on the right side.  It does not establish an analogous theorem for arbitrary degree or cut profiles, and it does not exclude intrinsic phenomena when the local endpoint width grows with the joint width.  The lifting family itself has growing common width but is computationally easy.  Accordingly, the paper establishes a new lower-bound measure and characterizes it in one fixed-degree regime; it does not establish a major algebraic complexity class separation.

One concrete consequence is nevertheless useful.  Since $C_{1,3}$ lower-bounds homogeneous ABP width, Corollary~\ref{cor:Fermat-lift} gives an explicit $\Omega(n)$ width lower bound for the easy family $F_n$.  More conceptually, the family shows that a common computation can require far more interface states than either interface requires when optimized separately, while Theorem~\ref{thm:completeness} shows that this discrepancy is fully witnessed by second derivatives when the degree-three local width is bounded.

\section{Conclusion}\label{sec:conclusion}

Two-cut coherence measures the cost of realizing two degree interfaces on one commutative tensor lift.  It can exceed both local restricted strengths, their sum, and their maximum by an unbounded factor, including after border closure.  For quintics at cuts $(1,3)$, however, this extra cost is not hidden: the largest cubic slice rank among square second derivatives controls it, up to constants whenever $C_3$ is bounded.  The main identity is the dimension-free sandwich
\[
  \left\lceil\frac{\Dtwo(f)}3\right\rceil
  \le\Cb{1,3}(f)
  \le C_{1,3}(f)
  \le C_3(f)\Dtwo(f)+2C_3(f)^2.
\]
This simultaneously supplies a border-stable lower-bound technique, an unbounded lifting separation, and a completeness theorem explaining the obstruction in the first nontrivial noncomplementary quintic regime.

\appendix

\section{Sharpness of the three-group transfer}\label{app:sharpness}

We prove that no universal bound $C_1(\partial_u^2H)\le Kr$ with $K<3$ can hold for all width-$r$ representations $H=p^{\mathsf T}Mq$.

\begin{lemma}[Disjoint-monomial cubic]\label{lem:disjoint-cubic}
Let
\[
  G_m=\sum_{s=1}^m x_sy_sz_s
\]
in $3m$ distinct variables.  Then $C_1(G_m)=m$.
\end{lemma}

\begin{proof}
The displayed expression gives the upper bound.  Suppose
\[
  G_m=\sum_{h=1}^q\ell_hQ_h.
\]
At every common zero of the $2q$ forms $\ell_1,\ldots,\ell_q,Q_1,\ldots,Q_q$, all first derivatives of $G_m$ vanish.  Hence this common zero locus is contained in $\Sing(G_m)$.  Krull's height theorem gives
\[
  \codim\Sing(G_m)\le2q.
\]
For each variable triple, the singular equations are
\[
  x_sy_s=x_sz_s=y_sz_s=0,
\]
whose solution is the union of three coordinate axes and has dimension one.  Over the $m$ disjoint triples, $\dim\Sing(G_m)=m$ in an ambient space of dimension $3m$, so $\codim\Sing(G_m)=2m$.  Therefore $q\ge m$.
\end{proof}

\begin{proposition}[Asymptotic sharpness]\label{prop:sharpness}
For every $r\ge2$, there is a width-$r$ quintic representation $H_r=p^{\mathsf T}Mq$ and a variable $\tau$ such that
\[
  C_1\left(\frac12\partial_\tau^2H_r\big|_{\tau=0}\right)=3r-1.
\]
\end{proposition}

\begin{proof}
Use independent linear variables
\[
  L_j,U_j,V_j,A_j,B_j,Y_j\quad(1\le j\le r)
\]
and
\[
  P_i,E_i,F_i\quad(2\le i\le r),
\]
together with $\tau$.  Put
\[
  p=(\tau,P_2,\ldots,P_r)^{\mathsf T}.
\]
Let the only nonzero entries of $M$ be
\[
  M_{1j}=A_jB_j+\tau L_j\quad(1\le j\le r),
  \qquad
  M_{i1}=E_iF_i\quad(2\le i\le r).
\]
Set
\[
  q_1=U_1V_1+\tau Y_1+\tau^2,
  \qquad
  q_j=U_jV_j+\tau Y_j\quad(2\le j\le r).
\]
Extracting the coefficient of $\tau^2$ gives
\[
  \frac12\partial_\tau^2(p^{\mathsf T}Mq)\big|_{\tau=0}
  =\sum_{j=1}^rL_jU_jV_j
   +\sum_{j=1}^rY_jA_jB_j
   +\sum_{i=2}^rP_iE_iF_i.
\]
The right side is a sum of $3r-1$ cubic monomials on disjoint triples.  Lemma~\ref{lem:disjoint-cubic} completes the proof.
\end{proof}

\section{Supplementary exact verification}\label{app:verification}

The supplementary archive contains two independent exact-arithmetic scripts.  The first uses symbolic polynomial expansion; the second uses custom sparse polynomial dictionaries over rational coefficients.  They verify:
\begin{enumerate}[leftmargin=2em]
  \item the three-group square-derivative identity~\eqref{eq:three-group};
  \item the four-group mixed-derivative identity mentioned after Corollary~\ref{cor:border-transfer};
  \item the lifting extraction identities $(\partial_a+\partial_b)^2L=2A$ and $(\partial_c+\partial_d)^2L=2B$;
  \item the coefficient-isolation identity~\eqref{eq:isolate-G};
  \item the $3r-1$ sharpness construction for several symbolic values of $r$.
\end{enumerate}
The scripts are consistency checks only.  The lower bounds, closedness statements, and completeness theorem are proved in the text and do not rely on computation.

\bibliographystyle{ACM-Reference-Format}
\bibliography{references}

\end{document}